\documentclass[a4,paper,11pt]{article}
\usepackage{booktabs}
\usepackage{epsfig}
\usepackage{graphicx}
\usepackage{rotating}
\usepackage{float}
\usepackage{subfigure}
\usepackage{natbib}
\usepackage{amsmath}
\allowdisplaybreaks
\usepackage{tikz}
\usepackage{multirow}
\headheight\baselineskip\headsep2\baselineskip\textwidth200mm
\advance\textwidth -2in\textheight287mm\advance\textheight-2in
\usepackage{graphicx}
\usepackage{epsfig}
\usepackage{amsfonts,amssymb,amsmath,lscape}
\usepackage{colortbl}
\usepackage{epsfig}
\usepackage{graphicx}
\usepackage{epstopdf}
\usepackage{amsthm}
\usepackage{lscape,color}

\definecolor{c30}{rgb}{0,0,1}

\newtheorem{Theorem}{Theorem}
\newtheorem{Definition}{Definition}

\newtheorem{Corollary}{Corollary}

\newtheorem{Example}{Example}
\newtheorem{Remark}{Remark}
\begin{document}

\title{\bf\Large \bf  A Unified Approach to Construct Correlation Coefficient Between Random Variables}
\author{
\textbf{ Majid Asadi\footnote{{Department of Statistics, University of Isfahan, Isfahan 81744, Iran \& School of  Mathematics, Institute of Research in Fundamental Sciences (IPM), P.O Box 19395-5746, Tehran, Iran} (E-mail: m.asadi@sci.ui.ac.ir)} \ \ and \ \ \textbf{Somayeh Zarezadeh\footnote{Department of Statistics, Shiraz University, Shiraz 71454, Iran  (E-mail: s.zarezadeh@shirazu.ac.ir) }}
}}
\date{}
\maketitle

\begin{abstract}
 Measuring the correlation (association)  between two random variables is one of the important goals in  statistical
 applications. In the literature, the covariance between two random variables  is a widely used criterion in measuring
  the linear association  between two random variables. In this paper,  first we propose a covariance based unified
  measure of variability for a  continuous random variable $X$ and we show that several measures of variability and
  uncertainty, such as variance, Gini mean difference, cumulative residual entropy, etc.,  can be considered as special
  cases.  Then, we propose a unified measure of correlation between two continuous random variables $X$ and $Y$, with
  distribution functions (DFs) $F$ and $G$,  based on the covariance between  $X$ and $H^{-1}G(Y)$ (known as the
  {\it Q-transformation} of $H$ on $G$)  where  $H$ is a continuous  DF.  We show that our proposed measure of   association subsumes some of the existing measures of correlation. {Under some mild condition on $H$}, it is shown the suggested index ranges  between  $[-1,1]$  where the extremes of the range, i.e., -1 and 1,  are  attainable by the    Fr$\acute{\rm e}$chet bivariate minimal and maximal DFs, respectively.
 { A special case of the proposed correlation measure leads to a variant of Pearson correlation coefficient which,
  as a  measure of strength and direction of the linear relationship between $X$ and $Y$, has absolute values greater
  than  or equal to the Pearson correlation.}  The results are examined numerically for some well known bivariate DFs.
\end{abstract}

\maketitle

\section{Introduction}\label{intro}

One of the fundamental issues in statistical theory and applications is to measure  the correlation (association)
between two random phenomena.
The problem of assessing the correlation  between two random variables (r.v.s) has a long history and because of
importance of the subject, several criteria  have been proposed in the statistical literature. Let $X$ and $Y$  be
two continuous  r.v.s  with joint distribution function (DF) $F(x,y)=P(X\leq x,Y\leq y)$,
$(x,y)\in \mathbb{R}^2$, and continuous marginal DFs $F(x)=P(X\leq x)$ and $G(y)=P(Y\leq y)$, respectively.
In parametric framework, the Pearson correlation coefficient,  which is the most commonly used type of correlation
index,  measures the strength and direction of the linear relationship between $X$ and $Y$. The Pearson correlation coefficient,
denoted by $\rho(X,Y)$, is defined as the ratio of the covariance between $X$ and $Y$,  to the product of their standard deviations. That is
\begin{equation}\label{prhop} \rho(X,Y)=\frac{\mathrm{Cov}(X,Y)}{\sigma_{X}\sigma_{Y}}=\frac{E(XY)-E(X)E(Y)}{\sigma_{X}\sigma_{Y}},
\end{equation}
where $\sigma_{X}>0$ $(\sigma_{Y}>0)$ denotes the standard deviation of $X$ $(Y)$.
An application of Cauchy-Schwarz inequality shows that $\rho(X,Y)$  lies in interval $[-1, 1]$.
In nonparametric framework, the widely used measures of association between two r.v.s  are Kendall's coefficient
and Spearman's coefficient. The Spearman correlation coefficient is defined as the Pearson correlation coefficient
 between the ranks of $X$ and $Y$ while the Kendall's coefficient (of concordance)  is expressed with respect to
 the probabilities of the concordant and discordant pairs of observations from $X$ and $Y$. For more information
 in  properties and applications of these indexes of correlation we refer, among others,  to \cite{Samuel et al. (2001),
  Shevlyakov and Oja (2016)} and references therein.
Although these correlation coefficients have  been widely used in many disciplines, there have been also defined
other indexes of associations which are particulary useful in certain areas of applications;  see, for example,
{ \cite{Yin (2004), R3, Nolde (2014), Grothe et al. (2014)}}.
In economic  and financial studies a commonly used measure of association between r.v.s $X$ and $Y$ is defined
based on Gini's mean difference  by \cite{R1}. The Gini's mean difference corresponding to  r.v. $X$,
denoted by $\mathrm{GMD}(X)$ (or alternatively with $\mathrm{GMD}(F)$),  is defined as
\begin{eqnarray}
  \mathrm{GMD}(X)= E(|X_1-X_2|)=2\int F(x)\bar{F}(x)dx,\label{gmd}
\end{eqnarray}
where $X_1$ and $X_2$ are independent r.v.s distributed as $X$ and $\bar{F}(x)=1-F(x)$. The GMD(X)  as  a measure
of  variability, (which is also equal to $4\mathrm{Cov}(X,F(X))$),  shares many properties of the variance of $X$
and  is more informative than the variance for the distributions that are far from normality (see, \cite{R2}).
\cite{R1} defined   the association between $X$ and $Y$  as the covariance between $X$ and $G(Y)$ divided by the
covariance between $X$ and $F(X)$. In other words, they proposed the measure of association between $X$ and $Y$ as
\begin{equation}
\Gamma(X,Y)=\frac{\mathrm{Cov}(X,G(Y))}{\mathrm{Cov}(X,F(X))}.\label{eee1}
\end{equation}
As  for a continuous r.v. $Y$, $G(Y)$ is distributed uniformly on $(0,1)$, the index  $\Gamma(X,Y)$ measures the
association between $X$ and a uniform r.v. on the interval $(0,1)$ which corresponds to the rank of $Y$. The
index $\Gamma(X,Y)$ has the requirements of a correlation coefficient and  is well applied in a series of
research works in economics and finance by Yitzhaki and his coauthors. We refer the readers, for more details
 on applications of $\Gamma(X,Y)$ and its extensions,  to  \cite{R3} and references therein.
 Recently, \cite{Asadi(2017)}  proposed a new measure of association between two continuous r.v.s $X$ and $Y$.
 This measure is defined on the basis of  $\mathrm{Cov}(X, \phi(X))$, where $\phi(x)=\log\frac{F(x)}{\bar{F}(x)},$
 is the log-odds rate associated to $X$. The cited author provides some interpretations of this covariance and showed
 that it arises  naturally as a measure of variability. For instance, it is shown that $\mathrm{Cov}(X, \phi(X))$ can
  be expressed as a function of cumulative residual entropy (a measure of uncertainty defined in \cite{R8}). Then the
    measure of association between r.v.s $X$ and $Y$ is defined as the ratio of the covariance between $X$ and the log-odds
     rate of  $Y$ divided by the covariance between $X$ and the log-odds rate of $X$. If we denote this measure by
     $\alpha(X,Y)$, then
 \begin{eqnarray}
   \alpha(X,Y)=\frac{\mathrm{Cov}(X, \phi_Y(Y))}{\mathrm{Cov}(X, \phi_{X}(X))}.\label{eqq2}
 \end{eqnarray}
  It should be noted that for  a continuous r.v. $X$, $\phi_{X}(X)$ is distributed as standard  Logistic distribution.
  Hence $\alpha(X,Y)$ measures the correlation between $X$ and a
  standard Logistic r.v.,  where the Logistic r.v. is the log-odds transformation of the r.v. $Y$.

 \qquad The aim of the present paper is to give a unified approach to construct
  measures of association between two r.v.s. In this regard,   we assume that $X$ and $Y$ have continuous
  DFs $F$ and $G$, respectively. First we consider the following covariance which we call it the $G$-covariance
   between $X$ and $Y$,
   \begin{equation}
    {\cal C}(X,Y)=\mathrm{Cov}\left(X,G^{-1}F(X)\right),\label{eq1}
  \end{equation}
  where, for $p\in[0,1]$,
  \[ G^{-1}(p)=\inf\{x\in \mathbb{R}: G(x)\geq p\},\]
is the inverse function of DF $G$. The quantity  $G^{-1}F(.)$ is known in the literature with different names.
\cite{Gilchrist(2000)} called it {\it Q-transformation}
(Q-T)  and \cite{Shaw and Buckley(2009)} named  it  {\it sample transmutation maps}. Throughout the paper, we use the
abbreviation Q-T for quantities  of the form $G^{-1}F(.)$.
Note that the covariance in (\ref{eq1}) measures the linear dependency between $X$ and r.v. $G^{-1}F(X)$, where
the latter one is a r.v. distributed as $Y$.  Based on the covariance  (\ref{eq1}), we propose a unified index of
correlation between $X$ and $Y$ which leads to new measures of correlations and subsumes some of the existing measures
such as the Pearson correlation coefficient (in the case that the $X$ and $Y$ are identical) and  Gini correlation coefficient (and its extensions).
Then, we study several properties of our unified index of association.

\qquad
The rest of the paper is organized as follows:  In Section 2, first we give briefly  some backgrounds and
applications
of quantity Q-T which have already presented in the literature. Then, we give the motivations of using the covariance
(\ref{eq1}) by showing that some measures of variability such as variance, GMD (and its extensions) and cumulative
residual entropy can be considered as special cases of (\ref{eq1}).
In Section 3, we propose our unified measure of association between the r.v.s $X$ and $Y$  based on the covariance
between $X$ and  $H^{-1}G(Y)$, where $H$ is a continuous DF. We call this unified correlation as  {\it $H$-transformed correlation} between $X$ and $Y$ and denote it by $\beta_H(X,Y)$. It is  shown that $\beta_H(X,Y)$ has almost all requirements of a correlation index. For example, it is proved that for any continuous  symmetric DF $H$,   $-1\leq \beta_H(X,Y)\leq 1$,
where  $\beta_H(X,Y)=0$ if $X$ and $Y$ are independent. When the joint distribution of
$X$ and $Y$ is bivariate normal with Pearson correlation $\rho(X,Y)=\rho$, we show that  $\beta_H(X,Y)=\rho$,  for
any $H$.  We prove that for the  association index  $\beta_H(X,Y)$  the  lower and upper bounds of the interval
$[-1,1]$ are attainable.  In fact, it is proved  that $\beta_H(X,Y)=-1$ $(+1)$ if $X$ and $Y$ are jointly  distributed
 as   Fr$\acute{\rm e}$chet bivariate minimal (maximal) distribution.  A special case of $\beta_H(X,Y)$, which we call  it {\it $\rho$-transformed  correlation} and denote it by $\rho_t(X,Y)$,  provides a variant of Pearson correlation
 coefficient $\rho(X,Y)$, whose absolute value is always greater than or equal to the absolute value of Pearson
 correlation $\rho(X,Y)$. That is, $\rho_t(X,Y)$ provides a wider range than that of $\rho(X,Y)$ for measuring the
 linear correlation between two r.v.s. The correlation  $\beta_H(X,Y)$ provides, in general,  an asymmetric class of
 correlation measures in terms of $X$ and $Y$. We propose some symmetric versions of that in Section 3.
 The index $\beta_H(X,Y)$  is computed for several bivariate distributions under different special cases for DF $H$.
 In Section 4,  a decomposition formula is given for G-covariance of sum of nonnegative r.v.s  which yields to some
 applications for redundancy systems. The paper is finalized with some concluding remarks in Section 5.

\section{Motivations}
Let $X$ and $Y$ be two continuous r.v.s with joint DF $F(x,y)$, $(x,y)\in \mathbb{R}^2$, and  marginal DFs $F(x)$ and
 $G(y)$, respectively. In developing our results the quantity Q-T, $G^{-1}F(x)$,  plays a central role.
 \cite{Balanda and MacGillivray(1990)} showed that the behavior of Q-T  can be used to assess the Kurtosis of two
 distributions (see, also, \cite{Groeneveld(1998)}). They showed that for symmetric distributions the so called
 {\it spread-spread } function is essentially a function of Q-T.  \cite{Shaw and Buckley(2009)} mentioned that
 among the applications of Q-T  is sampling from exotic distributions, e.g. $t$-Student. Authors have also used
 the plots of sample version of Q-T, in which the empirical distributions are replaced in $G^{-1}(F(x))$, for
 assessing symmetry of the distributions; see \cite{Doksum et al. (1977)} and references therein.
 \cite{Aly and Bleuer(1986)} called the function Q-T as the Q-Q plot and obtained some confidence intervals for that.
 In comparing the probability distributions, the concept of dispersive (variability) ordering  is used to measure
 variability of r.v.s (see, \cite{Shaked and Shanthikumar(2007)}). The  concept of dispersive ordering relies mainly
 on quantity $G^{-1}F(x)$. A DF $F$ is said to be less than a DF $G$  in dispersive ordering if $G^{-1}F(x)-x$
 is nondecreasing in $x$. (The dispersive ordering had  been  already employed  by  \cite{Doksum (1975)} in which
  he used the terminology \lq\lq $F$ is tail-ordered with respect to $G$\rq\rq).   \cite{Zwet(1964)}  used the
  quantity Q-T to compare the skewness of two probability density functions. The DF $G$ is more right-skewed,
  respectively more left-skewed, than
the DF $F$ if $G^{-1}F(x)-x$ is a nondecreasing convex, respectively concave, function (see also,
\cite{Yeo and Johnson(2000)}). In reliability theory, the convexity of the function Q-T is used, in a general
 setting,  to study the aging properties of lifetime r.v.s with support $[0,\infty)$
 (see, \cite{Barlow and Proschan(1981)}). In particular case if $G$ is exponential distribution,
 the convexity of Q-T  is equivalent to the property  that $F$ has increasing failure rate. Also, according to the
 latter cited authors, a lifetime DF $F$ is said to be less than a lifetime  DF $G$ in star-shaped order if
  $\frac{G^{-1}F(x)}{x}$ is increasing in $x$. In special  case that $G$ is exponential the star-shaped  property
  of Q-T  is equivalent to the property  that $F$ has increasing failure rate in average.

In the following, we use Q-T to  define a variant of covariance between  $X$ and $Y$  which we call it $G$-covariance.
Throughout the paper, we assume that all the required expectations exist.
\begin{Definition}
  {\rm Let $X$ and $Y$ be two r.v.s with DFs $F$ and $G$, respectively. The $G$-covariance of  $X$ in  terms of
   DF $G$ is defined as
  \begin{equation}
    {\cal C}(X,Y)=\mathrm{Cov}\left(X,G^{-1}F(X)\right).
  \end{equation}}
\end{Definition}
As $G^{-1}{F(x)}$ is an increasing function of $x$, we clearly have $0\leq \mathrm{Cov}\left(X,G^{-1}F(X)\right)$,
where the equality holds if and only if  $F$ (or G) is degenerate.  With $\sigma_X^2$ and $\sigma_Y^2$ as the
variances of $X$ and $Y$, respectively, using Cauchy-Schwarz inequality, we  have
\begin{eqnarray}
\mathrm{Cov}^{{2}}(X, G^{-1}F(X))&\leq&\mathrm{Var}(X)\mathrm{Var}{ (G^{-1}F(X))}\nonumber\\
                                   &=& \sigma^{2}_{X}\sigma^{2}_{Y}\label{covee}
\end{eqnarray}
 where the equality follows from the fact that $G^{-1}F(X)$ is distributed as $Y$. Hence, we get that
 \begin{eqnarray}
0\leq {\cal C}(X,Y)\leq \sigma_{X}\sigma_{Y}.\label{cove}
\end{eqnarray}
It can be easily shown that, in the right inequality of (\ref{cove}),
we have the equality if and only if  $X$ and $Y$ are distributed identically up to a location.

Note that ${\cal C}(X,Y)$ can be represented as
\begin{eqnarray}
  {\cal C}(X,Y)&=& \mathrm{Cov}(X,G^{-1}F(X))\nonumber\\
  &=& E\big(XG^{-1}F(X)\big)-E\big(G^{-1}F(X)\big)E(X)\nonumber\\
  &=& E\big(XG^{-1}F(X)\big)-E(Y)E(X)\\
  &=&\int xG^{-1}F(x)dF(x)-E(Y)E(X)\nonumber\\
   &=&\int yF^{-1}G(y)dG(y)-E(Y)E(X)\nonumber\\
  &=& \mathrm{Cov}(Y,F^{-1}G(Y))={\cal C}(Y,X).\label{eqe11}
\end{eqnarray}
Also an alternative way to demonstrate  ${\cal C}(X,Y)$ is
 \begin{eqnarray*}
  {\cal C}(X,Y)&=& \int xG^{-1}F(x)dF(x)-\int G^{-1}F(x)dF(x)\int xdF(x)\\
  &=& \int_{0}^{1} F^{-1}(u)G^{-1}(u)du-\int_{0}^{1} G^{-1}(u)du\int_{0}^{1} F^{-1}(u)du\\
  &=& \mathrm{Cov}(F^{-1}(U),G^{-1}(U)),
\end{eqnarray*}
where $U$ is a uniform r.v. distributed on $(0,1)$.

In the following we show that some well known measures of disparity and variability have a covariance representation and  can be considered as special cases of the $G$-covariance ${\cal C}(X,Y)$.
\begin{description}
  \item [(a)] If $G=F$, then we get
  $${\cal C}(X,Y)={\cal C}(X,X)=\mathrm{Cov}(X,F^{-1}(F(X)))=\mathrm{Cov}(X,X)=\mathrm{Var}(X).$$
  In particular if the vector $(X,Y)$ has an exchangeable DF then
  $${\cal C}(X,Y)=\mathrm{Var}(X)=\mathrm{Var}(Y)={\cal C}(Y,X).$$
  \item[(b)] If $G$ is uniform distribution on $(0,1)$, then we get
  $${\cal C}(X,Y)=\mathrm{Cov}(X,F(X))=\frac{1}{4}\mathrm{GMD}(X),$$
  where $\mathrm{GMD}(X)$ is the Gini's mean difference in (\ref{gmd}). The Gini coefficient, which is a widely
  used measure in economical studies, is defined as the $\mathrm{GMD}(X)$  divided by twice the mean of the
  population. It should be also noted that the $\mathrm{GMD}(X)$ can be represented as   the difference between
  the expected values of the maxima and the minima in a sample of two independent and identically distributed
  (i.i.d.) r.v.s $X_1$ and $X_2$. That is
\[
\mathrm{GMD}(X)= 4 \mathrm{Cov}(X, F(X)) = E \left(\max(X_1, X_2) - \min(X_1, X_2)\right);
\]
see, e.g.,  \cite{R3}.

In reliability theory and survival analysis, the mean residual life (MRL)  and mean inactivity time (MIT) are
important concepts to assess  the lifetime and aging  properties  of devices and live organisms. These concepts,
denoted respectively by $m(t)$ and ${\tilde m}(t)$, are defined at any time $t$ as $m(t)=E(X-t|X>t)$, and ${\tilde m}(t)=E(t-X|X<t)$.
Recently, \cite{R4} have shown, in the case that $X$ is a nonnegative r.v.,
$\mathrm{GMD}(X)$ (and hence $4{\cal C}(X,Y)$) can  also be expressed as the sum of expectations of MRL
and MIT of the minimum of random sample of size 2.
\item[(c)] In the case that $G(y)=1-e^{-y}$, $y>0$, the exponential distribution with mean 1, we obtain
  $${\cal C}(X,Y)=\mathrm{Cov}(X,\Lambda(X)),$$
  where $\Lambda(x)=G^{-1}F(X)=-\log\bar{F}(x)$,  in which $\bar{F}(x)=1-F(x)$.
  The function  $\Lambda(x)$, corresponding to a nonnegative r.v.,  is called in reliability theory as the
  cumulative failure rate and  plays a crucial role in the study of aging properties of systems lifetime.  \cite{Asadi(2017)} has shown that the following equality holds for a nonnegative r.v.
  \begin{equation}\label{CRE}
  \mathrm{Cov}(X,\Lambda(X))=-\int_{0}^{\infty}\bar{F}(x)\log\bar{F}(x)dx,
  \end{equation}
  where the right hand side is known, in the literature, as the cumulative residual entropy (CRE) defined by \cite{R8}.
  As an alternative measure of Shannon entropy,  the cited authors  argued that CRE can be considered as a measure of uncertainty.  They  obtained several properties of CRE and  illustrated that this measure is useful in computer vision and image processing. \cite{R16} showed that the CRE is closely related  to the mean residual life, $m(t)$, of a nonnegative r.v. $X$. If fact, it is always true that the CRE can be represented as  $\mathrm{CRE}=E(m(X))$.
  Another interesting  fact that can also be concluded from the discussion here is that the differential Shannon entropy of the equilibrium distribution (ED) corresponding to $F$ has a covariance representation. The density function of ED is given by
\[ f_e(x)= \frac{\bar{F}(x)}{\mu},\]
where $0<\mu<\infty$ is the mean of DF $F$. In a renewal process, the ED arises as the asymptotic distribution of the waiting time until the next renewal and the time since the last renewal at time $t$. Also a delayed renewal process has stationary
increments if and only if the distribution of the actual
remaining life is $f_{e}(x)$. Such process  known in the literature as the stationary
renewal process or equilibrium renewal process; see, \cite{Ross}. If $H(f_e)$ denotes the differential Shannon entropy of $f_e$, then
\begin{eqnarray*}
  H(f_e)&=& -\int_{0}^{\infty}f_{e}(x)\log f_{e}(x)dx\\
  &=& -\int_{0}^{\infty}\frac{\bar{F}(x)}{\mu}\log \frac{\bar{F}(x)}{\mu}dx\\
  &=& \frac{1}{\mu}\mathrm{Cov}(X,\Lambda(X))+\log \mu.
\end{eqnarray*}

Finally,  we should mention in this part, that the concept of generalized cumulative residual entropy (GCRE) which is  introduced by \cite{Psar-Nav} as
\begin{align}
{\cal E}_n(X)= \frac{1}{n!} \int_0^\infty {\bar F}(x) [\Lambda(x)]^n dx. \label{psar}
\end{align}
For $n=1$, we get the CRE of $X$. One can easily verify that, with $G_n(y)=1-e^{-\sqrt[n]{y}}$,  ${\cal E}_n(X)$ has the following covariance representation
\begin{align}
{\cal E}_n(X)=\frac{1}{n!} \mathrm{Cov}\big(X,G^{-1}_{n}F(X)\big)-\frac{1}{(n-1)!} \mathrm{Cov}\big(X,G^{-1}_{n-1}F(X)\big).\label{gre2}
\end{align}

\item [(d)] In the case that $G$ is Logistic with DF $G(y)=\frac{1}{1+e^{-y}}$, $y\in \mathbb{R}$, we obtain
   $${\cal C}(X,Y)=\mathrm{Cov}\left(X,\phi(X)\right),$$
where $\phi(x)=\log\frac{F(x)}{\bar{F}(x)},$ is  the log-odds rate associated to r.v. $X$.
Log-odds rate is considered  in the survival analysis to model the failure process of lifetime data to assess the survival function of observations (see, \cite{R5}).
It is easy to show that
   \begin{eqnarray*}
   {\cal C}(X,Y)&=&\mathrm{Cov}\left(X,\phi(X)\right)\\
      &=& -\int_{0}^{\infty}\bar{F}(x)\log\bar{F}(x)dx-\int_{0}^{\infty}{F}(x)\log{F}(x)dx,
   \end{eqnarray*}
   where the last term on the right hand side is called as the cumulative past entropy. For some discussions and interpretations of ${\cal C}(X,Y)$, presented in this part, see Asadi (2017).
\item[(e)] Let
\begin{align}\label{exten-gini-dis}
  G(y)=\left\{
    \begin{array}{ll}
      1-\Big(\frac{1}{y}\Big)^{\frac{1}{1-\nu}}, & \hbox{$y>1, ~ 0<\nu<1$; {\rm \ \ Pareto distribution,}}\\
      1-(1-y)^\frac{1}{\nu-1}, & \hbox{$0<y<1, ~ \nu>1$; {\rm\ \  Power distribution,} } \\
      0, & \hbox{o.w.}
    \end{array}
  \right.
\end{align}
Then it can be shown, in this case, that
\begin{equation*}
     {\cal C}(X,Y)=\left[I(0< \nu < 1)-I(\nu>1)\right]\mathrm{Cov}(X,\bar{F}^{\nu-1}(X)),
\end{equation*}
where $I(A)$ is an indicator function which is equal to 1 when $x\in A$ and otherwise is equal to zero. Hence, we get the extended Gini, $\mathrm{EGini}_\nu(X)$,  defined as a parametric extension of GMD(X) of the form:
\[
\mathrm{EGini}_\nu(X)=\nu \left[I(\nu>1)-I(0< \nu < 1)\right] {\cal C}(X,Y),
\]
where $\nu$ is a parameter ranges from 0 to infinity and  determines  the relative weight attributed to various portions of probability distribution.  For  $\nu=2$, the extended Gini leads to GMD(X) (up to a constant).  For more interpretations and applications  of  $\mathrm{EGini}_\nu(X)$ in economic studies based on  different values of $\nu$, we refer to  \cite{R3}.

\item[(f)]   The upper and lower record values, in a sequence of i.i.d. r.v.s $X_1, X_2, \dots$, have applications in different areas of applied probability; see, \cite{Arnold et al. (1998)}. Let $X_i$'s have a common continuous DF $F$ with survival function ${\bar F}$. Define a sequence of upper record times $U(n)$, $n = 1, 2, \dots$, as follows
$$
U(n + 1) = \min \{j : j > U(n), X_j > X_{U(n)}\}, \quad n \geq 1,
$$
with $U(1) = 1$. Then, the sequence of upper record values $\{R_n, n\geq 1\}$ is defined
by $R_n= X_{U(n)}$, $n\geq 1$, where $R_1 = X_1$. The survival function of $R_n$ is given by
\[
{\bar F}_n^U(t)={\bar F}(t)\sum_{x=0}^{n-1} \frac{(\Lambda(t))^x}{x!}, \qquad t>0, n=1,2,\dots,
\]
where $\Lambda(t)=-\log{\bar F}(t)$.   If $R_n$ denotes the $n$th upper record value,  then it can be easily shown that,  with $G_n(y)=1-e^{-\sqrt[n]{y}}$, the mean of difference between $R_n$  and $R_1$ has the following covariance representation:
\begin{align*}
 E(R_n-R_1)=&E(R_n-\mu)=\frac{1}{(n-1)!} \mathrm{Cov}\big(X,G_{n-1}^{-1}F(X)\big),\ \ \ n\geq 1,
\end{align*}
where $\mu=E(R_1)=E(X_1)$.

The lower record values in a sequence of i.i.d. r.v.s $X_1, X_2, \dots$ can be defined in a similar manner. The sequence of record times $L(n)$, $n = 1, 2, \dots$, is defined as $L(1) = 1$ and
\[ L(n + 1) = \min \{j : j > L(n), X_j < X_{L(n)}\}, \quad n\geq1. \]
Then the $n$th lower record value is defined by ${\tilde R}_n=X_{L(n)}$. The DF of ${\tilde R}_n$ is given by
\[ F_n^L(t)=F(t)\sum_{x=n}^{\infty} \frac{[{\tilde \Lambda}(t)]^x}{x!}, \quad t > 0,~n = 1, 2,\dots, \]
in which ${\tilde \Lambda}(t)=-\log F(t)$; see, \cite{Arnold et al. (1998)}.

 Let ${\tilde R}_n$ denote the $n$th lower record. Then, it can be shown that
\begin{align*}
 E({\tilde R}_n-{\tilde R}_1)= E({\tilde R}_n-\mu)=\frac{1}{(n-1)!} \mathrm{Cov}\big(X,[{\tilde \Lambda}(X)]^{n-1}\big),\ \ \ n\geq 1,
\end{align*}
 where ${\tilde \Lambda}(t)=-\log F(t)$. Therefore the expectation of the difference between the $n$th upper and lower records has a covariance representation as follows
\begin{align*}
E({R}_n-{\tilde R}_n)=&\frac{1}{(n-1)!} \mathrm{Cov}\big(X,[{\Lambda}(X)]^{n-1}\big)-\frac{1}{(n-1)!} \mathrm{Cov}\big(X,[{\tilde \Lambda}(X)]^{n-1}\big)\\
=& \frac{1}{(n-1)!} \mathrm{Cov}\big(X,K_n^{-1}(F(X))\big),
\end{align*}
where $K_n(x)$ is a DF with inverse $K_n^{-1}(u)=(-\ln(1-u))^n-(-\ln(u))^n$, $0<u<1$.
\end{description}

\section{A Unified Measure of Correlation}
 We define our unified measure of correlation between $X$ and $Y$,  as follows:
 \begin{Definition}\em
 Let $X$ and $Y$ be   two continuous  r.v.s  with joint DF $F(x,y)$, $(x,y)\in \mathbb{R}^2$, and continuous marginal DFs $F(x)$ and $G(y)$, respectively. Let $H$ be a  continuous DF.  Then the  $H$-transformed correlation between $X$ and $Y$, denoted by $\beta_H(X,Y)$,  is defined as
 \begin{eqnarray}
   \beta_H(X,Y)=\frac{\mathrm{Cov}(X, H^{-1}G(Y))}{\mathrm{Cov}(X, H^{-1}F(X))},\label{index}
 \end{eqnarray}
 provided that all expectations exist and  $\mathrm{Cov}(X, H^{-1}F(X))>0$.
 \end{Definition}
 It is trivial that for continuous r.v. $Y$, the r.v. $H^{-1}G(Y)$ is distributed as r.v. $W$, where $W$ has DF $H$. Hence, $\beta_H(X,Y)$ measures the association between $X$ and a function of $Y$ where that function is the transformation $H^{-1}$ over $G(Y)$. The  $H$-transformed correlation between $Y$ and $X$ can be defined similarly as
 \begin{eqnarray*}
   \beta_H(Y,X)=\frac{\mathrm{Cov}(Y, H^{-1}F(X))}{\mathrm{Cov}(Y, H^{-1}G(Y))},
 \end{eqnarray*}
 provided that $\mathrm{Cov}(Y, H^{-1}G(Y))>0$.

 In what follows, we study the properties of $\beta_H(X,Y)$ and show that, {under some mild condition on $H$}, it has  the necessary requirements of a correlation coefficient. Before that, we give the following corollary  showing  that $\beta_H(X,Y)$ subsumes some well known measures of association as special cases.
\begin{Corollary}
  {\rm  The correlation  index  $\beta_H(X,Y)$ in (\ref{index}) gives the following measures of association as special cases:

\begin{description}
  \item [(a)]  { If we assume that $H= G$ then we have
        \begin{eqnarray}
   \beta_H(X,Y)
   &=& \frac{\mathrm{Cov}(X, Y)}{\mathrm{Cov}(X, G^{-1}F(X))}\nonumber\\
    &=& \frac{\mathrm{Cov}(X, Y)}{\mathrm{Cov}^{\frac{1}{2}}(X, G^{-1}F(X))\mathrm{Cov}^{\frac{1}{2}}(Y, F^{-1}G(Y))},\label{rhort}
 \end{eqnarray}
 where the last equality follows from (\ref{eqe11}). In the following,  we call (\ref{rhort}) as
 {\it $\rho$-transformed correlation} between $X$ and $Y$ and denote it by $\rho_t(X,Y)$.
 The measure $\rho_t(X,Y)$ is a correlation index   proportional to the Pearson correlation coefficient
 $\rho(X,Y)$ in (\ref{prhop}).  In fact
 $\rho_t(X,Y)= a \rho(X,Y),$
 where
 \[a=\frac{\sigma_X\sigma_Y}{\mathrm{Cov}^{\frac{1}{2}}(X, G^{-1}F(X))\mathrm{Cov}^{\frac{1}{2}}(Y, F^{-1}G(Y))}.\]
 In particular, if the marginal DFs  $F$ and $G$ are identical, then  $a=1$.
 (Note that, a sufficient condition to have $F=G$ is that the joint DF of $(X,Y)$ to be exchangeable.
  Recall that a random vector $(X,Y)$ is said to have an exchangeable DF if the vectors  $(X,Y)$ and $(Y,X)$ are identically distributed.)
 However, in general case  based on (\ref{covee}),  we   always have
 $$\mathrm{Cov}^{{2}}(X, G^{-1}F(X))=\mathrm{Cov}^{{2}}(Y, F^{-1}G(Y))\leq \sigma^{2}_{X}\sigma^{2}_{Y}.$$
 Hence, we get that $a^2\geq1$. This, in turn,  implies that the  following interesting inequality holds between
 $\rho(X,Y)$ and $\rho_t(X,Y)$:
 \begin{equation}\label{erho}
  { 0\leq |\rho(X,Y)|\leq |\rho_t(X,Y)|.}
  \end{equation}
 {We will show in Theorem \ref{th1q} that when $X$ and $Y$ are positively correlated then $\rho(X,Y)\leq \rho_t(X,Y)\leq 1$, and when $X$ and $Y$ are negatively correlated and $G$ or $F$ is a symmetric DF,  then $-1\leq \rho_{t}(X,Y)\leq \rho(X,Y)$.   These inequalities indicate  that  $\rho_t(X,Y)$, as a measure of the strength and direction of the linear  relationship between two r.v.s, in compare  to the Pearson correlation $\rho(X,Y)$, shows  more intensity  of  correlation  between the two r.v.s.. This may be  due to the fact that in denominator of $\rho(X,Y)$   the normalizing factor  $\sigma_X$ $(\sigma_Y)$ depends only  on the distribution of $F$ $(G)$  while in   denominator of $\rho_t(X,Y)$ the normalizing factor $\mathrm{Cov}^{\frac{1}{2}}(X, G^{-1}F(X))$   $(\mathrm{Cov}^{\frac{1}{2}}(Y, F^{-1}G(Y)))$ depends on  both DFs  $F$ and $G$.
 }}

\item  [(b)]  If $H$ is uniform on interval $(0,1)$, i.e., $H(x)=x$, $0<x<1$,  then  $\beta_H(X,Y)$ reduces to
the Gini correlation in (\ref{eee1}),
         $$\Gamma(X,Y)=\frac{\mathrm{Cov}(X,G(Y))}{\mathrm{Cov}(X,F(X))}.$$

\item [(c)] { If   $H$ is Pareto distribution $(0<\nu<1)$ or power distribution $(\nu>1)$, given in below
         \begin{align*}
  H(x)=\left\{
    \begin{array}{ll}
      1-\Big(\frac{1}{x}\Big)^{\frac{1}{1-\nu}}, & \hbox{$x>1, ~ 0<\nu<1$;}\\
      1-(1-x)^\frac{1}{\nu-1}, & \hbox{$0<x<1, ~ \nu>1$;}\\
    0, & \hbox{o.w.,} \end{array}
  \right.
\end{align*}
we get the extended Gini $(\mathrm{EGini_{\nu}})$ correlation defined as
\[ \Gamma(\nu, X,Y)=\frac{\mathrm{Cov}(X,\bar{G}^{\nu-1}(Y))}{\mathrm{Cov}(X,\bar{F}^{\nu-1}(X))}, \quad \nu>0.\]
Note that for $\nu=2$ we arrive at the Gini correlation.}

\item [(d)] If $H(x)=\frac{1}{1+e^{-x}}$, $x\in \mathbb{R}$, the standard Logistic distribution, then $\beta_H(X,Y)$ becomes the association measure in (\ref{eqq2}), defined  by \cite{Asadi(2017)}, which measures the correlation between $X$  and the log-odds rate of  $Y$.

  \end{description}}
\end{Corollary}

\qquad Before giving the main properties of the correlation in (\ref{index}), we give the  following expressions which indicate that the correlation coefficient $\beta_H(X,Y)$ has representations in terms of joint DF $F(x,y)=P(X\leq x,Y\leq y)$  and joint survival function $\bar{F}(x,y)=P(X>x,Y>y)$. In the sequel, we assume that all  the  integrals  are  from $-\infty$ to $\infty$ unless  stated  otherwise. The correlation $\beta_{H}(X,Y)$ can be expressed as
     \begin{align*}
    \beta_H(X,Y)=&\frac{1}{\mathrm{Cov}(X,H^{-1}F(X))}\int\int\left({F}(x,y)-{F}(x){G}(y)\right)dxdH^{-1}G(y)\\
     =&\frac{1}{\mathrm{Cov}(X,H^{-1}F(X))}\int\int\left(\bar{F}(x,y)-\bar{F}(x)\bar{G}(y)\right)dxdH^{-1}G(y).
     \end{align*}

 The validity of these expressions  can be verified  from Theorem 1 of \cite{R9} under the  assumptions that the expectations exist and $H^{-1}G(y)$ is a bounded variation function.

The following theorem gives some properties of ${\beta}_H(X,Y)$.
\begin{Theorem}\label{th1q}
 The correlation ${\beta}_H(X,Y)$ satisfies in the following properties:
\begin{description}
  \item [(a)]  { For  continuous r.v.s $X$ and $Y$, ${\beta}_H(X,Y)\leq 1$ and when $H$ is a symmetric DF, $-1\leq {\beta}_H(X,Y)\leq 1$.}
    \item [(b)] The maximum (minimum) value of ${\beta}_H(X,Y)$  is achieved, if $Y$ is a monotone increasing (decreasing) function of $X$.
      \item[(c)] For independent r.v.s $X$ and $Y$, ${\beta}_H(X,Y)={\beta}_H(Y,X)=0$.
      \item[(d)] ${\beta}_H(X,Y)=-{\beta}_H(-X,Y)=-{\beta}_H(X,-Y)={\beta}_H(-X,-Y).$
      \item[(e)] The correlation  measure ${\beta}_H(X,Y)$ is invariant  under all strictly monotone functions of $Y$.
      \item[(f)]  ${\beta}_H(X,Y)$ is invariant under changing the location and scale of $X$ and $Y$.
      \item [(g)] If the joint DF of $X$ and $Y$ is exchangeable, then  ${\beta}_H(X,Y)={\beta}_H(Y,X)$.
\end{description}
\end{Theorem}

\begin{proof}
We provide the proofs for parts (a) and (g).  The proofs of other parts are straightforward (see, \cite{R3}, p. 41, where the authors study the properties of Gini correlation $\Gamma(X,Y)$).
\begin{description}
\item[(a)] {First, we show that ${\beta}_H(X,Y)\leq 1$ for any continuous DF $H$. To this,  We need to show that  $E(XH^{-1}G(Y))\leq E(XH^{-1}F(X)).$
Both functions $X$ and $H^{-1}F(X)$ are increasing functions. Then $E(XH^{-1}G(Y))$ achieves its maximum value
when $H^{-1}G(Y)$ is an increasing function of $X$, (see \cite{R3}, p. 41).
This implies that  $H^{-1}F(X)=H^{-1}G(Y)$ which, in turn,  implies that the maximum value is achieved at
$E(XH^{-1}F(X))$ and hence  ${\beta}_H(X,Y)\leq 1$.

Now, let $H$ be a symmetric DF about constant $a$. To have $-1\leq {\beta}_H(X,Y)$ it needs to show that
$-\mathrm{Cov}(X,H^{-1}F(X))\leq \mathrm{Cov}(X,H^{-1}G(Y)).$
From  \cite{R3}, p. 41,
$E(XH^{-1}G(Y))$ achieves its minimum value
when $H^{-1}G(Y)$ is a decreasing function of $X$.
This results in $H^{-1}G(Y)=H^{-1}(1-F(X))=2a-H^{-1}(F(X))$ which, in turn, implies that
$2a-E(XH^{-1}F(X))\leq E(XH^{-1}G(Y))$ and hence $-\mathrm{Cov}(X,H^{-1}F(X))\leq \mathrm{Cov}(X,H^{-1}G(Y))$. Hence, we have $-1\leq {\beta}_H(X,Y)$.
}
\item[(g)] As the  random vector $(X,Y)$ has exchangeable distribution,  $(X,Y)$ is identically distributed as
$(Y,X)$ and hence the marginal distributions of $X$ and $Y$ are identical, i.e.,  $F=G$. Hence, we can write
\begin{align*}
{\beta}_H(X,Y)=&~\frac{\mathrm{Cov}(X,H^{-1}G(Y))}{\mathrm{Cov}(X,H^{-1}F(X))}\\
              =&~\frac{\mathrm{Cov}(X,H^{-1}F(Y))}{\mathrm{Cov}(X,H^{-1}G(X))}\\
              =&~\frac{\mathrm{Cov}(Y,H^{-1}F(X))}{\mathrm{Cov}(Y,H^{-1}G(Y))}={\beta}_H(Y,X).
\end{align*}
\end{description}
\end{proof}

The following theorem proves that in bivariate normal distribution, the correlation $\beta_{H}(X,Y)$ is
equal to Pearson correlation $\rho(X,Y)$.

\begin{Theorem}
  Let $X$ and $Y$ have   bivariate normal distribution with Pearson correlation coefficient $\rho(X,Y)=\rho$.
  Then, for any continuous DF $H$ with finite mean $\mu_{H}$,  $${\beta}_H(X,Y)={\beta}_H(Y,X)=\rho.$$
\end{Theorem}
\begin{proof}
Assume that the marginal DFs of $X$ and $Y$ are $F$ and $G$, with means $\mu_F$ and $\mu_G$ and positive variances
$\sigma^2_{F}$ and $\sigma^2_{G}$, respectively. Further let $Z$ denote the standard normal r.v. with DF $\Phi$.
It is well known that for the bivariate normal distribution we have
$$E(X|Y)=\mu_F +\rho \sigma_F\frac{(Y-\mu_G)}{\sigma_G}.$$
Using this we can write
\begin{align*}
\mathrm{Cov}(X,H^{-1}G(Y))=&~E_Y\left[\left(E(X|Y)-\mu_{F}\right)\left(H^{-1}G(Y)-\mu_{H}\right)\right]\\
=&~ \rho\sigma_{F} E_Y\Big[\big(\frac{Y-\mu_G}{\sigma_G}\big)H^{-1}G(Y)\Big]\\
=&~\rho\sigma_F \int \big(\frac{y-\mu_G}{\sigma_G}\big) H^{-1}G(y) dG(y)\\
=&~\rho\frac{\sigma_F}{\sigma_G} \int \big({G^{-1}\Phi(z)-\mu_G}\big) H^{-1}\Phi(z) d\Phi(z)\\
=& ~ \rho\frac{\sigma_F}{\sigma_G} \big( \int G^{-1}\Phi(z) H^{-1}\Phi(z) d\Phi(z)-\mu_G\mu_H\big)\\
=&~\rho\frac{\sigma_F}{\sigma_G} \mathrm{Cov}(G^{-1}\Phi(Z),H^{-1}\Phi(Z))\\
=&~\rho\sigma_{F} \mathrm{Cov} (Z, H^{-1}\Phi(Z)),
\end{align*}
where the last equality follows from the fact that $G^{-1}\Phi(z)=\sigma_G z+\mu_G$.
On the other hand, we can similarly show that $\mathrm{Cov}(X,H^{-1}F(X))=\sigma_{F}\mathrm{Cov}(Z,H^{-1}\Phi(Z)).$
Hence. we have  ${\beta}_H(X,Y)=\rho.$
\end{proof}

 \qquad Assuming that $X$ and $Y$ have joint  bivariate DF $F(x,y)$,  with marginal DFs $F(x)$ and $G(y)$, then $F(x,y)$ satisfies the  Fr$\acute{\rm e}$chet bounds inequality
 $$F_0(x,y)=\max\{F(x)+G(y)-1,0\}\leq F(x,y)\leq  \min\{F(x),G(y)\}=F_1(x,y).$$
 The  Fr$\acute{\rm e}$chet bounds $F_0(x,y)$ and $F_1(x,y)$ are themselves bivariate distributions known as the minimal and maximal distributions, respectively. These distributions show  the perfect negative and positive  dependence between the corresponding r.v.s   $X$ and $Y$, respectively; in the sense that \lq\lq  the joint distribution of $X$ and $Y$ is $F_0(x,y)$ ($F_1(x,y)$) if and only if  $Y$ is decreasing (increasing) function of $X$\rq\rq\ (see, \cite{R13}).
  In the following theorem we prove that, under some conditions,  the  extremes of the range for ${\beta}_H(X,Y)$
 i.e., $-1$ and $1$,  are  attainable by the   Fr$\acute{\rm e}$chet bivariate minimal and maximal distributions,
 respectively.
 In other words, we show that for  lower and upper  bounds of  Fr$\acute{\rm e}$chet inequality
 we have ${\beta}_H(X,Y)=-1$ and ${\beta}_H(X,Y)=1$, respectively.

 \begin{Theorem}\label{Fer}
        Let $X$ and $Y$ be two continuous r.v.s with DFs ${F}(x)$ and ${G}(y)$, respectively, and $H$ be a continuous DF.
        \begin{itemize}
          \item[(a)]  If $(X,Y)$ has joint DF  $F_1(x,y)$ then ${{\beta}_H}(X,Y)=1$,
          \item[(b)]  If $H$ is symmetric and $(X,Y)$ has joint DF $F_0(x,y)$ then ${{\beta}_H}(X,Y)=-1$.
          \end{itemize}
      \end{Theorem}
\begin{proof}%
\begin{description}
      \item[(a)] Let us define the sets  $A_x=\{y|y\geq{G}^{-1}({F}(x))\}$ and $A^c_x=\{y|y<{G}^{-1}({F}(x))\}$.
                Then, we have
      \begin{align}
   \mathrm{Cov}(X,H^{-1}G(Y))=&\int\int\left({F}(x,y)-F(x)G(y)\right)dH^{-1}G(y)dx\nonumber\\
    =& \int\int\Big(\min\{{F}(x),{G}(y)\}-F(x)G(y)\Big)dH^{-1}G(y)dx\nonumber\\
      =&\int{F}(x)\int_{A_x}\bar{G}(y)dH^{-1}G(y)dx+\int\bar{F}(x)\int_{A^c_x}{G}(y)dH^{-1}G(y)dx. \label{eqf}
                     \end{align}
         But, we have under the assumptions of the theorem
         \begin{equation}
           \int_{A_x}\bar{G}(y)dH^{-1}G(y)=-\bar{F}(x)H^{-1}F(x)+\int_{F(x)}^{1}H^{-1}(u)du,\label{eqf1}
         \end{equation}
         and
         \begin{equation}
           \int_{A^c_x}{G}(y)dH^{-1}G(y)=F(x)H^{-1}F(x)-\int_{0}^{F(x)}H^{-1}(u)du.\label{eqf2}
             \end{equation}
From (\ref{eqf}), (\ref{eqf1}) and (\ref{eqf2}), we get
             \begin{align}
   \mathrm{Cov}(X,H^{-1}G(Y))=&\lim_{a\rightarrow -\infty}\Big\{\int_{a}^{\infty} F(x)\int_{F(x)}^{1}H^{-1}(u)dudx-\int_{a}^{\infty}\bar{F}(x)\int_{0}^{F(x)}H^{-1}(u)dudx\Big\}\nonumber\\
   =&\lim_{a\rightarrow -\infty}\Big\{\int_{a}^{\infty} F(x)\int_{F(x)}^{1}H^{-1}(u)dudx\nonumber \\
   &-\int_{a}^{\infty}\bar{F}(x)\Big(\int_{0}^{1}H^{-1}(u)du-
   \int_{F(x)}^{1}H^{-1}(u)du\Big)dx\Big\}\nonumber\\
   =&\lim_{a\rightarrow -\infty}\Big\{ \int_{a}^{\infty}\int^{1}_{F(x)}H^{-1}(u)dudx-\int_{a}^{\infty}\bar{F}(x)dx \int_{0}^{1}H^{-1}(u)du\Big\}\nonumber\\
   =& \lim_{a\rightarrow -\infty}\Big\{\int_{0}^{1}H^{-1}(u)\Big(\int_a^{F^{-1}(u)} dx -\int_a^\infty {\bar F}(x)dx\Big)du\Big\}\nonumber\\
   =& \lim_{a\rightarrow -\infty}\Big\{\int_{0}^{1}H^{-1}(u)(F^{-1}(u)-a-\mu_F+a\big)\Big\}\nonumber\\
   =&\int_{0}^{1}F^{-1}(u)H^{-1}(u)du-\mu_F\mu_H\nonumber\\
   =&~\mathrm{Cov}(F^{-1}(U),H^{-1}(U))\nonumber\\
   =&~\mathrm{Cov}(X,H^{-1}(F(X)))\nonumber.
                        \end{align}
                        This shows that $\beta_{H}(X,Y)=1$.
\item[(b)] In this case we define  $B_x=\{y|y\geq{G}^{-1}(\bar{F}(x))\} $ and $B^c_x=\{y|y<{G}^{-1}(\bar{F}(x))\}.$
Then
 \begin{align*}
   \mathrm{Cov}(X,H^{-1}G(Y))=& \int_{a}^{\infty}\int\Big(\max\{{F}(x)+{G}(y)-1,0\}-F(x)G(y)\Big)dH^{-1}G(y)dx\\
      =&-\int_{a}^{\infty}\int_{B_x}\bar{F}(x)\bar{G}(y)dH^{-1}G(y)dx-\int\int_{B^c_x}F(x)G(y)dH^{-1}G(y)dx.
   \end{align*}
Therefore, using the same procedure as part (a), it can be written
\begin{align*}
\mathrm{Cov}(X,H^{-1}G(Y))=&\lim_{a\rightarrow -\infty}\Big\{\int_{a}^{\infty}F(x)\int_{0}^{\bar{F}(x)}H^{-1}(u)dudx-
      \int_{a}^{\infty}\bar{F}(x)\int_{\bar{F}(x)}^{1}H^{-1}(u)dudx\Big\}\\
      =& \lim_{a\rightarrow -\infty}\Big\{ \int_{a}^{\infty}\int_{0}^{\bar{F}(x)}H^{-1}(u)dudx
      -\int_{a}^{\infty}\bar{F}(x) dx \int_{0}^{1}H^{-1}(u)du\Big\}\\
      =& \lim_{a\rightarrow -\infty}\Big\{ \int_0^1 H^{-1}(1-u) \big(F^{-1}(u)-a-\mu_F+a\big)du\\
      \stackrel{c}{=}& \int_0^1 (2\mu_H-H^{-1}(u)) F^{-1}(u) du-\mu_F\mu_H\\
      =& -\int_0^1 H^{-1}(u) F^{-1}(u) du+\mu_F\mu_H\\
      =&-\mathrm{Cov}(H^{-1}(U),F^{-1}(U))\\
      =&-\mathrm{Cov}(X,H^{-1}(F(X))),
\end{align*}
       where the equality ($c$) follows from the assumption that $H$ is symmetric. Hence, we get that $\beta_{H}(X,Y)=-1$. This completes the proof of the theorem.
\end{description}
\end{proof}
\begin{Remark}\em
{ It should be pointed out that,  the symmetric condition imposed on $H$  in part (b) of Theorem
  \ref{Fer} can not be dropped in general case. As a counter example, it can be easily verify  that if $H$ is
   exponential the upper bound 1 for $\beta_{H}(X,Y)$ is attainable by Fr$\acute{\rm e}$chet bivariate maximal
   distribution, however, the lower bound -1 is not attainable by Fr$\acute{\rm e}$chet bivariate minimal
   distribution.}
  \end{Remark}

\qquad A well known class of bivariate distributions, which is extensively studied in the statistical literature,
is FGM family (see, \cite{R10}). The joint DF $F(x,y)$ of the r.v.s  $X$ and $Y$  with, respectively,
continuous marginal DFs $F(x)$ and $G(y)$, is said to be a member of FGM family if
$$F(x,y)=F(x)G(y)\left(1+\gamma\bar{F}(x)\bar{G}(y)\right), $$
where $\gamma\in[-1,1]$ shows  the parameter of dependency between $X$ and $Y$. Clearly for $\gamma=0$, $X$ and $Y$ are
independent. It is well known  that for FGM family
the  Pearson correlation coefficient $\rho(X,Y)$ lies  in interval $[-1/3,1/3]$ where the maximum is attained for the
case when the marginal distributions are uniform (\cite{R11}).
\cite{R12} proved  that in FGM family, the Gini correlation $\Gamma(X,Y)$  lies between $[-{1}/{3},{1}/{3}],$ for any
marginal DFs $F$ and $G$.

The following theorem gives an expression for ${\beta}_H(X,Y)$ in FGM family.
\begin{Theorem}
  Under the assumption that $F$, $G$ and $H$ have finite means, the association measure ${\beta}_H(X,Y)$, for the FGM  class, is given by
  \begin{eqnarray}\label{betaFGM}
 {\beta}_H(X,Y)=\gamma\frac{\mathrm{GMD}(F)\mathrm{GMD}(H)}{4\mathrm{Cov}(X,H^{-1}F(X))}.
  \end{eqnarray}
  \end{Theorem}
\begin{proof}
\begin{align*}
    \mathrm{Cov}(X,H^{-1}G(Y))=& \int\int\left({F}(x,y)-{F}(x){G}(y)\right)dH^{-1}G(y)dx\\
    =&~\gamma\int\int F(x)\bar{F}(x)G(y)\bar{G}(y)dH^{-1}G(y)dx\\
    =&~\gamma\int F(x)\bar{F}(x)dx \int G(y)\bar{G}(y)dH^{-1}G(y)\\
    =&~\gamma\int F(x)\bar{F}(x)dx\int H(u)\bar{H}(u))du\\
    =&\frac{\gamma}{4}\mathrm{GMD}(F) \mathrm{GMD}(H),
\end{align*}
   where $\bar{H}=1-H$. Hence,  ${\beta}_H (X,Y)$  can be represented as
\begin{eqnarray*}
  {\beta}_H(X,Y)=\gamma\frac{\mathrm{GMD}(F) \mathrm{GMD}(H)}{4\mathrm{Cov}(X,H^{-1}F(X))}.
 \end{eqnarray*}
This completes the proof.
\end{proof}

{It should be pointed out  that the correlation index  ${\beta}_H(X,Y)$ in FGM family does not depend on the DF $G$
which is transmuted  by $H$. Also, it is trivial that, in the case where $H$ is uniform DF on interval $(0,1)$,
${\beta}_H(X,Y)$ reduces to Gini correlation which is free of $F$ and its values lies in $[-1/3,1/3]$.
If $H=G$,   we arrive at the following formula for  $\rho_t(X,Y)$:
\begin{eqnarray*}
  {\rho}_t(X,Y) &=&\frac{\gamma}{4}\frac{\mathrm{GMD}(F)}{\mathrm{Cov}^{1/2}(X,G^{-1}F(X))} \frac{\mathrm{GMD}(G)}{\mathrm{Cov}^{1/2}(Y,F^{-1}G(Y))}.
 \end{eqnarray*}
Table \ref{tabrhotfgm} gives the range of possible values of $\rho(X,Y)$ and  $\rho_{t}(X,Y)$, in FGM family,
 for different choices of  DFs $F$ and $G$.
When one of the two r.v.s  is selected as uniform r.v. $U$, then we get the Gini correlation and hence
\[
\rho_t(U,X)=\rho_t(X,U)=\frac{\mathrm{GMD}(X)\mathrm{GMD}(U)}{4\mathrm{Cov}(X,F(X))}=\frac{\gamma}{3}.
\]
This implies that the range of possible values of $\rho_t(X,U)$ is $[-1/3,1/3]$.
As seen in the table, $\rho_t(X,Y)$, in compare  to the Pearson correlation $\rho(X,Y)$, shows a wider range
of correlation  between the two r.v.s.

\begin{table}[!h]
\centering
\caption{\small The ranges of $\rho$, and $\rho_t$ correlations for some distributions in FGM family.}\label{tabrhotfgm}
\small
\begin{tabular}{lccccccc}
   \toprule
    & & \multicolumn{5}{c}{DF of $Y$} \\ \cline{3-7}
     DF of $X$ & Index &
    Uniform & Exponential & Reighley & Logistic & Normal \\[2mm]
    \hline\hline
    Uniform & $\rho(X,Y)$ & $\mp0.33333$ & $\mp0.28867$ & $\mp 0.32352$ & $\mp 0.31831$ & $\mp 0.32573$\\
    & $\rho_t(X,Y)$ & $\mp0.33333$ & $\mp0.33333$ & $\mp0.33333$ & $\mp0.33333$ & $\mp0.33333$ \\ \hline
    Exponential & $\rho(X,Y)$ & $\mp 0.28867$ & $\mp 0.25000$ & $\mp 0.28016$ &  $\mp 0.27566$        &   $\mp 0.28209$     \\
    & $\rho_t(X,Y)$ & $\mp 0.33333$ & $\mp 0.25000$ & $\mp 0.29289$ &  $\mp 0.30396$      &    $\mp 0.31233$   \\
    \hline
    Reighley & $\rho(X,Y)$ & $\mp 0.32352$  & $\mp 0.28016$   & $\mp0.31396$  &  $\mp 0.30892$  & $\mp 0.31613$    \\
             & $\rho_t(X,Y)$ & $\mp 0.33333$  & $\mp 0.29289$  & $\mp0.31396$  &  $\mp 0.31549$  & $\mp 0.32057$
    \\    \hline
    Logistic & $\rho(X,Y)$   &  $\mp 0.31831$ & $\mp 0.27566$ & $\mp 0.30892$ & $\mp 0.30396$ & $\mp 0.31105$\\
             & $\rho_t(X,Y)$ &  $\mp 0.33333$ & $\mp 0.30396$ & $\mp 0.31549$ & $\mp 0.30396$ & $\mp 0.31233$ \\ \hline
    Normal   & $\rho(X,Y)$   &  $\mp 0.32573$ & $\mp 0.28209$ & $\mp 0.31613$ & $\mp 0.31105$ & $\mp 0.31831$\\
           & $\rho_t(X,Y)$   &  $\mp 0.33333$ & $\mp 0.31233$ & $\mp 0.32057$ & $\mp 0.31233$ & $\mp 0.31831$ \\
    \bottomrule
    \end{tabular}
\end{table}
}

In the following, we give some examples in which $\beta_{H}(X,Y)$  in (\ref{index}) are computed for different transformation DFs $H$.
The following choices for $H$ are considered:
\begin{itemize}
  \item   Exponential distribution $H(x)=1-e^{-x},~ x>0$: Cumulative residual entropy based (CRE-Based) correlation.
  \item Logistic distribution, $H(x)= \frac{1}{1+e^{-x}},~x\in \mathbb{R}$: Odds ratio based (OR-Based) correlation.
\item  Pareto distribution, $H(x)=1-\Big(\frac{1}{x}\Big)^{2},~ x>1$: Extended Gini correlation with parameter $\nu=0.5$ ($\mathrm{EGini}_{0.5})$.
 \item   Uniform distribution, $H(x)=x,~~0<x<1$: Gini correlation.
 \item  Power distribution, $H(x)= 1-(1-x)^{\frac{1}{2}}, ~ 0<x<1$: Extended Gini correlation with parameter $\nu=3$ ($\mathrm{EGini}_3$).
\end{itemize}
\begin{Example}
{\rm Table \ref{tablefgm} represents the values of $\beta_{H}(X,Y)$, in FGM family,   for different choices of
transformation DFs $H$ and different DFs  $F$.}
\end{Example}
\begin{table}[!h]
\caption{\small The range of $\beta_{H}(X,Y)$  for different choices of $H$ and $F$ in FGM family.}\label{tablefgm}
\small
\centering
\begin{tabular}{llcccc}
\toprule
\multicolumn{5}{r}{The ranges of correlation coefficients} \\
\cline{3-6}
 Distribution   & $F(x)$ & CRE-Based & OR-Based & EGini$_{0.5}$ & EGini$_{3}$ \\
\midrule\midrule
 Weibull (1,0.5)             &$1-e^{-\sqrt{x}},~ x>0$ &  $\mp 0.18750$ & $\mp 0.26344$ & $\mp 0.08333$ & $\mp 0.42187$\\[2mm]
Exponential (1)                  &$1-e^{-x},~ x>0$    &  $\mp 0.25000$ & $\mp 0.30396$ & $\mp 0.16667$ & $\mp 0.37500$\\[2mm]
Weibull (1,2)              &$1-e^{-x^2},~ x>0$      & $\mp 0.29289$ & $\mp 0.31549$ & $\mp 0.23570$  & $\mp 0.34650$\\[2mm]
Logistic (0,1)  & $(1+e^{-x})^{-1}, ~ x\in \mathbb{R}$ &$\mp 0.30396$ & $\mp 0.30396$ & $\mp 0.24045$  & $\mp 0.33333$\\[2mm]
Extreme value (0,1)    &$~~e^{-e^{-x}}, ~ x\in \mathbb{R}$  &  $\mp 0.27555$ & $\mp 0.30701$ & $\mp 0.19951$ & $\mp 0.35335$\\
Laplace (0,1)  &$\left\{
   \begin{array}{ll}
     \frac{1}{2}e^{x}, & \hbox{$x<0$;} \\
     1-\frac{1}{2}e^{-x}, & \hbox{$x\geq 0$.}
   \end{array}
 \right.
 $ & $\mp 0.29403$ & $\mp 0.29403$ & $\mp 0.21832$ & $\mp 0.33333$\\
\bottomrule
\end{tabular}

\end{table}

\begin{Example}\label{example}\em
In this example we consider two bivariate distributions and compute the correlation index $\beta_{H}(X,Y)$ for different choices of $H$:
\begin{itemize}
\item [(a)]
The  first bivariate distribution which we consider is a special case of Gumbel-Barnett family of copulas,  introduced by \cite{Barnett(1980)}, given as
\begin{align}\label{copula}
C_\theta(u, v) = u + v - 1 + (1-u)(1-v)e^{-\theta \log(1-u)\log(1.v)}, \qquad 0\leq \theta \leq 1.
\end{align}
In this copula if we take the standard exponential DFs as marginals of $X$ and $Y$, then we arrive at  the Gumbel's bivariate exponential DF (Gumbel 1960). The joint DF of Gumbel's bivariate exponential distribution is written as
\begin{align}
F_\theta(x,y)=1-e^{-x}-e^{-y}+e^{-x-y-\theta xy},\qquad x>0,~y>0,~0\leq\theta\leq1. \label{gumbel}
\end{align}
For $\theta = 0$, $X$ and $Y$ are independent and $\rho(X,Y) = 0$.
As $\theta$ increases, the absolute value of Pearson correlation, $|\rho(X,Y)|$, increases and takes value $\rho(X,Y)=-0.40365$ at $\theta =1$.
This distribution is applied for describing r.v.s with negative correlation. (Of course, positive correlation can be obtained by changing $X$ to $-X$ or $Y$ to $-Y$.)  In Table \ref{table-1}, the range of Pearson correlation and the range of $H$-transformed correlation are given for Gumbel's bivariate exponential distribution.

\item [(b)]
The second bivariate distribution considered  in Table \ref{table-1} is bivariate Logistic distribution which is belong to Ali-Mikhail-Haq family of copulas (\cite{Hutchinson and Lai(1990)}) with the following structure
\begin{align}
C_\theta(u,v)=\frac{uv}{1-\theta(1-u)(1-v)}, \qquad -1\leq \theta \leq 1. \label{ali-Mik-copula}
\end{align}
With  standard Logistic distributions as marginal DFs of $X$ and $Y$,  we arrive at the joint DF of bivariate Logistic distribution as follows
\begin{align}\label{LogDF}
  F_\theta(x,y)=\dfrac{1-e^{-x}}{1+e^{-y}-\theta e^{-y-x}},\qquad x>0,~y\in \mathbb{R},~-1\leq \theta\leq1.
\end{align}
Note that Gumbel's bivariate Logistic distribution is a special case of bivariate Logistic distribution when $\theta = 1$.
\end{itemize}

Both bivariate DFs in \eqref{gumbel} and \eqref{LogDF} are exchangeable.  Hence for both cases, we obtain
$\rho(X,Y)=\rho_t(X,Y)$. The range of possible values of $\beta_{H}(X,Y)$ is given on the basis of five different
DFs $H$ introduced above. For each $H$, the values  of lower bound and upper bound of $H$-transformed
correlation for Gumbel's bivariate exponential, which are attained in $\theta=1$ and $\theta=0$, respectively,
are given in the first panel of Table \ref{table-1}.  It is seen from the table that the widest  range of
correlation is achieved for $\mathrm{EGini}_3$  among all other correlations.
It is evident from the table that, the range of the values of  Pearson correlation $\rho(X,Y)$ is even
less than those of  Gini and OR-based correlations. The minimum range of correlation corresponds to
 $\mathrm{EGini}_{0.5}$. In the case that the DF $H$ is equal to the marginal DFs of the bivariate distribution,
 the associated correlation $\beta_{H}(X,Y)$  becomes the Pearson correlation, which in this case is
 the CRE-Based correlation.
The second panel of Table \ref{table-1} gives the correlation $\beta_{H}(X,Y)$, based on the above  mentioned
distributions $H$, in   bivariate Logistic distribution. The lower bound and the upper bound of all correlations
are attained for  $\theta=-1$ and $\theta=1$, respectively.  In this case  the maximum range of correlation is achieved for $\mathrm{EGini}_3$
and the  minimum range is achieved for $\mathrm{EGini}_{0.5}$.
\begin{table}[!h]
\centering
\caption{\small The ranges of $\rho$, and $\beta_H$ correlations for two exchangeable distributions.}\label{table-1}
\small
\begin{tabular}{ lcc }
\toprule
\multicolumn{3}{ l}{{\bf Gumbel's Type I Bivariate  Exponential Distribution}} \\[2mm]
\multicolumn{3}{ l }{$F_\theta(x,y)=1-e^{-x}-e^{-y}+e^{-x-y-\theta xy},~~x>0,~y>0,~0\leq\theta\leq1.$} \\[2mm]
\hline
 Correlation index & Lower bound  & Upper bound  \\ \hline
 Pearson  & $-0.40365$ & $0$ \\
 CRE-Based & $-0.40365$ & $0$\\
 OR-Based & $-0.51267$ & $0$ \\
 $\mathrm{EGini}_{0.5}$ & $-0.26927$ & $0$ \\
 $\mathrm{Gini}$ & $-0.55469$ & $0$ \\
 $\mathrm{EGini}_3$ & $-0.64125$ & $0$ \\
\hline\hline
\multicolumn{3}{l}{{\bf Bivariate Logistic Distribution}} \\[2mm]
\multicolumn{3}{l}{$F_\theta(x,y)=\left(1+e^{-x}+e^{-y}+(1-\theta)e^{-x-y}\right)^{-1},~~~x\in \mathbb{R},~y\in \mathbb{R},~-1\leq \theta\leq1.$} \\[2mm]
\hline
 Correlation index & Lower bound  & Upper bound \\ \hline
 Pearson  & $-0.25000$ & $0.50000$\\
 CRE-Based & $-0.26516$ & $0.39207$\\
 OR-Based & $-0.25000$ & $0.50000$ \\
 EGini$_{0.5}$ & $-0.22135$ & $0.27865$ \\
 Gini & $-0.27259$ & $0.50000$  \\
 EGini$_{3}$ & $-0.26272$ & $0.55556$ \\
\bottomrule
\end{tabular}
\end{table}
\end{Example}

\begin{Example}\label{example-non-exchangeable}\em
  In this example, we  consider again the copulas given in  \eqref{copula} and \eqref{ali-Mik-copula}.
  However, here we assume that the marginal DFs are not the same (the bivariate distribution is not exchangeable).
  In the first  bivariate distribution the marginals  are two different Weibull DFs (with different shape parameters)
   and in the second  case the  marginals are two different power DFs (with different shape parameters), respectively.
    In Table \ref{table-2}, the ranges of possible values of $\rho(X,Y)$, $\rho_t(X,Y)$, and $\beta_{H}(X,Y)$ are
    presented for both bivariate DFs.
The values  of lower and upper bounds of $H$-transformed correlation for the two bivariate distributions which
are attained in $\theta=1$ and $\theta=0$, and in $\theta=-1$ and $\theta=1$, respectively, are numerically computed
for different DFs $H$. In the first panel which corresponds to Gumbel-Barnett copula with Weibull-Weibull marginals,
it is seen that   the maximum range is attained for $\mathrm{EGini}_3$ and the minimum range is achieved for Pearson correlation.
   Also as we showed in inequality (\ref{erho}), the results of the table show that  the $\rho$-transformed correlation     has a wider range than that of Pearson correlation.

 The second panel of the table presents the correlations between $X$ and $Y$ for
 Ali-Mikhail-Haq copula with power-power marginal DFs. %
   In this case, we see that the maximum range coincides with $\mathrm{EGini}_3$, the next maximum ranges are related
   to  OR-Based, and Gini correlations, respectively,  and the minimum range is obtained in $\mathrm{EGini}_{0.5}$.
   { Also we see that $\rho_{t}(X,Y)$  indicates  a wider range of correlation between $X$ and $Y$ comparing to Pearson  correlation $\rho(X,Y)$.}

\begin{table}[!h]
\centering
\caption{\small The ranges of $\rho$, $\rho_t$, and $\beta_H$ correlations for  two distributions with non-equal marginals.}\label{table-2}
\small
\begin{tabular}{ llcc }
\toprule
\multicolumn{3}{ l}{{\bf Gumbel-Barnett copula with Weibull-Weibull  marginals}} \\[2mm]
\multicolumn{3}{ l }{$F_\theta(x,y)=1-e^{-x^2}-e^{-\sqrt{y}}+e^{-x^2-\sqrt{y}-\theta x^2\sqrt{y}},\quad x>0,~y>0,~0\leq\theta\leq1.$} \\[2mm]
\hline
Correlation index & Lower bound  & Upper bound  \\ \hline
   Pearson & $-0.32420$ & $0$\\
   $\rho$-transformed   &    $-0.43307$ & $0$\\
  CRE-Based & $-0.48426$ & $0$ \\
  OR-Based & $-0.51759$ & $0$ \\
  EGini$_{0.5}$ & $-0.41563$ & $0$ \\
  Gini &  $-0.53692$ & $0$ \\
  EGini$_{3}$ & $-0.55776$ & $0$ \\
\hline\hline
\multicolumn{3}{ l}{{\bf Ali-Mikhail-Haq copula with power-power marginals}} \\[2mm]
\multicolumn{3}{ l }{
$F_\theta(x,y)=\dfrac{x(2-x)y(y^2-3y+3)}{(1+\theta(y-1)^3(x-1)^2)}, ~~~~0<x<1,~0<y<1,~-1\leq \theta\leq1$}\\[3mm]
\hline
Correlation index & Lower bound  & Upper bound  \\ \hline
 Pearson & $-0.27099$ & $0.39668$\\
 $\rho$-transformed  & $-0.27212$ & $0.39833$\\
 CRE-Based & $-0.26589$ & $0.36447$ \\
 OR-Based &  $-0.27387$ & $0.45685$\\
 EGini$_{0.5}$ & $-0.24790$ & $0.29890$\\
 Gini & $-0.27887$ & $0.45177$  \\
 EGini$_{3}$ & $-0.28324$ & $0.51025$\\
 \bottomrule
\end{tabular}
\end{table}
\end{Example}

\subsection{Some Symmetric Versions }
We have to point out here that   the Pearson's and Spearman's correlation
coefficients  are both  symmetric measures of correlation. However the association measure $\beta_H(X,Y)$ introduced
in this paper is not  generally a  symmetric measure  unless the two r.v.s are exchangeable. There are several ways
that one can introduce a symmetric version  of the correlation coefficient considered in this paper, i.e., to impose
a correlation coefficient with the property $\beta_{H}(X,Y) = \beta_{H}(Y,X)$. Motivated by the works of
\cite{R7, Yitzhaki and Olkin(1991), R2}, in the following,  we introduce three  measures of correlation based
on $\beta_{H}(X,Y)$ which are symmetric in terms of $F$ and $G$.
\begin{description}
 \item [(a)] The first symmetric version of correlation can be considered as
  \begin{equation}
  \tau_{H}(X,Y)=\frac{1}{2}\left(\beta_{H}(X,Y)+\beta_{H}(Y,X)\right).
  \end{equation}

  \item[(b)] The second symmetric version which can be constructed is based on the approach used by \cite{R2}. Let
      $\eta_X=Cov(X,H^{-1}F(X))$ and $\eta_Y=Cov(Y,H^{-1}G(Y))$.  Define  $\nu_H(X,Y)$ as follows
$$\nu_{H}(X,Y)=\frac{\eta_X\beta_{H}(X,Y)+\eta_Y\beta_{H}(Y,X)}{\eta_X+\eta_Y}.$$
Then $\nu_{H}(X,Y)$, as a  weighted function of $\beta_H(X,Y)$ and $\beta_H(Y,X)$,  is a symmetric measure of correlation that lies between $[-1,1]$ and  have  the requirements of a correlation coefficient described in  Theorem \ref{th1q}.
\item [(c)] The third symmetric index which can be imposed based on $\beta_{H}(X,Y)$ is as follows (see, \cite{R7}).
With  $\eta_X$, and $\eta_Y$, as defined in (b), let  ${\bar \beta}_H(X,Y)=1-\beta_{H}(X,Y)$ and
${\bar \beta}_H(Y,X)=1-\beta_{H}(Y,X).$  Consider  ${\bar \nu}_H(X,Y)$ as
\begin{align*}
  {\bar \nu}_{H}(X,Y)=&\frac{\eta_X{\bar \beta}_H(X,Y)+\eta_Y{\bar \beta}_H(Y,X)}{\eta_X+\eta_Y}\\
  =& 1-{\nu}_{H}(X,Y).
\end{align*}
Then ${\bar \nu}_{H}(X,Y)$ which is a  weighted function of ${\bar \beta}_H(X,Y)$ and ${\bar \beta}_H(Y,X)$
is symmetric in $F$ and $G$ and ranges between $[0,2]$. \cite{R7} showed that  ${\bar \nu}_{H}(X,Y)$, in the
case that $H$ is uniform distribution  gives a measure,  called Gini index of mobility,
that provides  a consistent  setting for analysis of mobility, inequality and horizontal equity.  It can
be easily shown  that ${\bar \nu}_{H}(X,Y)$  can be also presented  as
$${\bar \nu_{H}}(X,Y)=\frac{Cov\left(X-Y,H^{-1}F(X)-H^{-1}G(Y)\right)}{\eta_X+\eta_Y}.$$
\end{description}
In the following, we give an example that these symmetric measures are calculated.
\begin{Example}\label{examp-sym}
\em
Consider the Gumbel-Barnett copula with two different Weibull distributions as marginals and the joint DF given in
Example \ref{example-non-exchangeable}. Let $\theta=1$ which corresponds to highest  dependency between r.v.s $X$,
and $Y$. Then the joint DF of $X$ and $Y$ is written as
\[
F(x,y)=1-e^{-x^2}-e^{-\sqrt{y}}+e^{-x^2-\sqrt{y}- x^2\sqrt{y}},\qquad x>0,~y>0.
\]
Table \ref{table-sym} presents   the values of correlations $\beta_H(X,Y)$ and $\beta_H(Y,X)$,  symmetric
correlations  $\tau_H(X,Y)$, $\nu_H(X,Y)$ and ${\bar \nu}_H(X,Y)$ for different distributions  $H$.

\begin{table}[!h]
\centering
\caption{\small The values of symmetric correlation coefficients for Example \ref{examp-sym}.}\label{table-sym}
\small
\begin{tabular}{ llcccc }
\toprule
Index & $\beta_H(X,Y)$  & $\beta_H(Y,X)$  &  $\tau_H(X,Y)$ & $\nu_H(X,Y)$ & ${\bar \nu}_H(X,Y)$ \\
\midrule
    CRE-Based & $-0.48426$ & $-0.29817$  & $-0.39121$ & $-0.31673$ & $1.31673$\\
  OR-Based & $-0.51759$ & $-0.47762$  & $-0.49761$ & $-0.48267$ & $1.48267$\\
  EGini$_{0.5}$ & $-0.41563$ & $-0.12179$ & $-0.26871$ & $-0.13873$ & $1.13873$\\
  Gini &  $-0.53692$ & $-0.59375$  & $-0.56534$ & $-0.58537$ & $1.58537$\\
  EGini$_{3}$ & $-0.55776$ & $-0.80720$ & $-0.68248$ & $-0.76379$ & $1.76379$\\
\bottomrule
\end{tabular}
\end{table}
\end{Example}

\section{A Decomposition Formula}
 In this section we give a decomposition formula  for  ${\cal C}(T,Y)$, which provides some results on the connection
 between the variability of sum of a number of r.v.s in terms of sum of variabilities of each r.v.
 In a reliability  engineering point of view, consider a system with standby components with the following structure.
 We assume that the  system is  built of $n$  units with lifetimes $X_1,\dots,X_n$  which will be  connected to
 each other sequentially as follows. Unit number 1 with lifetime $X_1$ starts operating and in the time of  failure,
 the unit number 2 with lifetime $X_2$ starts working  automatically,  and so on until the $n$th unit, with lifetime
 $X_n$, fails. Hence,  the lifetime of the system, denoted by $T$,  would be $T=\sum_{i=1}^{n}X_i$.
Assume that $\mu_i=E(X_i)$ denotes the mean time to failure of unit number $i$ and $\mu=E(T)=\sum_{i=1}^{k}\mu_i$
denotes the mean time to failure of the system.

Let again for any two r.v.s $X$ and $Y$ with DFs $F$ and $G$, respectively, we denote
$\mathrm{Cov}\left(X,G^{-1}F(X)\right)$ by ${\cal C}(X,Y)$. Now we have the following result.
\begin{Theorem}\label{th-decomps}
For any r.v. $Y$ with DF $G$, we have the following decomposition for  ${\cal C}(T,Y)$ in terms of ${\cal C}(X_i,Y)$,
$i=1,2,\dots,n$.
\begin{eqnarray*}
{\cal C}(T,Y)=\sum_{i=1}^{n} {\beta}_G(X_i,T){\cal C}(X_i,Y),
\end{eqnarray*}
where $ {\beta}_G(X_i,T)$ is the $G$-transformed correlation between the system lifetime $T$ and component lifetime
$X_i$ defined in (\ref{index}).
\end{Theorem}
\begin{proof} Let $F_{X_i}$ and $F_T$ denote the DFs of component lifetime $X_i$ and the system lifetime $T$,
respectively. From the covariance properties of sum of r.v.s, we can  write
      \begin{eqnarray*}
      {\cal C}(T,Y)&=& Cov(T,G^{-1}F_T(T))\\
        &=&\sum_{i=1}^{n}Cov(X_i,G^{-1}F_T(T))\\
        &=&\sum_{i=1}^{n}\frac{Cov(X_i,G^{-1}F_T(T))}{Cov(X_i,G^{-1}F_{X_i}(X_i))}Cov(X_i,G^{-1}F_{X_i}(X_i))\\
        &=&\sum_{i=1}^{n}\beta_G(X_i,T){\cal C}(X_i,Y).
      \end{eqnarray*}
\end{proof}

\begin{Corollary}\label{cor1}
\em It is interesting to note that the correlation between the system lifetime $T$ and its component lifetime $X_i$, i.e., $\beta_G(X_i,T)$ is always nonnegative. This is so because in $ \beta_G(X_i,T)$, $G^{-1}F_T(T)$ is  trivially an increasing function of $X_i$, as $T$ is increasing function of $X_i$. Hence, $\mathrm{Cov}(X_i,G^{-1}F_T(T))$ is nonnegative which, in turn, implies  that  $ \beta_G(X_i,T)$ is nonnegative.
Thus,  we have
\begin{align}\label{inequal-alpha}
  0\leq \beta_G(X_i,T)\leq 1.
\end{align}
This result shows that the G-covariance between the system lifetime $T$ and  r.v. $Y$ can be decomposed as a combination of the G-covariance between components lifetime and r.v. $Y$.
 From Theorem \ref{th-decomps} and relation \eqref{inequal-alpha}, we conclude that
\begin{eqnarray}\label{inequality}
      {\cal C}(T,Y)\leq \sum_{i=1}^{n} {\cal C}(X_i,Y).
\end{eqnarray}
That is,  the G-covariance between the system lifetime and r.v. $Y$ is less than the sum of G-covariance between its components and r.v. $Y$.  In particular when the $X_i$'s are identical r.v.s, we have $ {\cal C}(T,Y)\leq n  {\cal C}(X_1,Y)$. In this situation, if we assume that $G=F_{X_1}$ then $${\cal C}(T,X_i)\leq n  \mathrm{Var}(X_1), \qquad i=1,\dots,n.$$
\end{Corollary}

Based on Corollary  \ref{cor1}, the following  inequalities are obtained for some well known measures of disparity as special cases:
\begin{itemize}
\item [(a)] If $G=F_T$, then we get
  $$\mathrm{Var}(T)\leq \sum_{i=1}^n {\cal C}(X_i,T).$$
\item[(b)] In the case that $G(x)=1-e^{-\sqrt[k]{x}}, \ x>0, \ k>0,$ the Weibull distribution with shape parameter $1/k$, we obtain
  $${\cal E}_k(T)\leq \sum_{i=1}^n {\cal E}_k(X_i), \qquad k=1,2,\dots,$$
  where the ${\cal E}_k(\cdot)$ is the GCRE defined in (\ref{psar}).  In  the special case where $k=1$, we obtain the following inequality regarding CRE.
  $${\cal E}_1(T)\leq \sum_{i=1}^n {\cal E}_1(X_i).$$
Thus, it is concluded that the uncertainty of  a stand by system lifetime, in the sense of  CRE, is less than the sum of uncertainties  of the   its components lifetime.  As a result we can also conclude equivalently that for the system described above
$$E(m_{T}(T))\leq \sum_{i=1}^{n} E(m_{X_i}(X_i)),$$
where $m_T$ and $m_{X_i}$ are the MRL's of the system and the components, respectively; see also, \cite{nasr}.
\item[(c)] Consider $G(\cdot)$ as the DF given in \eqref{exten-gini-dis}. Then, for $\nu>0$,
\begin{align*}
\mathrm{EGini}_\nu(T)\leq \sum_{i=1}^n \mathrm{EGini}_\nu(X_i).
\end{align*}
 For $\nu=2$, which corresponds to $G(x)$ as uniform distribution on $(0,1)$,  we get
  \begin{align}\label{inequality-Gini}
    \mathrm{GMD}(T)\leq \sum_{i=1}^n \mathrm{GMD}(X_i),
  \end{align}
  where $\mathrm{GMD}(\cdot)$ is the Gini's mean difference. This result was already obtained by \cite{R3}.
\end{itemize}
\section{Concluding Remarks}
In the present article, we introduced a unified  approach to construct  a  correlation coefficient between two
continuous r.v.s. We assumed that the continuous r.v.s  $X$ and $Y$ have a joint distribution function $F(x,y)$
with marginal distribution functions  $F$ and $G$, respectively.  We first considered the  covariance between $X$
and transformation $G^{-1}F(X)$, i.e., $\mathrm{Cov}(X,G^{-1}F(X))$. The function $G^{-1}F(.)$ is known in the
literature as the {\it Q-transformation} (or {\it sample transmutation maps}). We showed that  some well known
measures of variability such as variance, Gini mean difference and its extended version, cumulative residual entropy and
some other  disparity  measures can be considered as special cases of $\mathrm{Cov}(X, G^{-1}F(X))$. Motivated by this,
we proposed a unified measure of correlation between the r.v.s $X$ and $Y$  based on $\mathrm{Cov}(X,H^{-1}G(Y))$, where $H$ is a continuous distribution function.
We showed that the introduced measure, which subsumes some well known measures of associations such as Gini and Pearson correlations for special choices of $H$,  has all  requirements of a correlation index {under some mild condition on DF $H$}. For example it was shown that
it lies between $[-1, 1]$. When the joint distribution of $X$ nd $Y$ is bivariate normal, we showed that the proposed measure, for any choice of $H$, equals the Pearson correlation coefficient.  We proved, under some conditions that for our unified association index,    the  lower and upper bounds of the interval  $[-1, 1]$  are attainable by  joint Fr$\acute{\rm e}$chet bivariate minimal  and maximal distribution functions, respectively. A special case of the introduced correlation in this paper, provided  a variant of Pearson correlation coefficient $\rho(X,Y)$, which measures  with the property that its absolute value is always greater than or equal to the  absolute value of $\rho(X,Y)$.   Since  the proposed measure of correlation is asymmetric, some symmetric versions of that were also discussed. Several examples of bivariate DFs of $X$ and $Y$ were presented in which the correlation is computed for different choices of $H$. Finally, we  presented  a decomposition     formula for $\mathrm{Cov}(X,G^{-1}F(X))$ in which the r.v. $X$ was considered as  the sum of $n$  r.v.s. As an application of the decomposition formula,  some  results were  provided  on the  connection between variability   measures of a standby system in terms of the  variability measures of its components.

The r.v.s that we considered in this article, were assumed to be continuous. One interesting problem which can be
considered as a future study is to investigate the results for the case that  the r.v.s are arbitrary
(in particular discrete r.v.s). Another important problem which can be investigated is to propose some estimators
for $\beta_{H}(X,Y)$ for different choices of $H$.
In particular, we believe that providing estimators for $\rho_{t}(X,Y)$ and exploring their  properties may
be of special importance,  for measuring the linear correlation between the real data collected in different
disciplines and  applications.

\end{document}